\documentclass[11pt,letter]{article}

\usepackage{amsmath, amsthm, amssymb}
\usepackage{verbatim}
\usepackage[sort]{natbib}
\usepackage{xcolor}
\usepackage[pagebackref=true,colorlinks]{hyperref}
\hypersetup{linkcolor=blue,filecolor=blue,citecolor=blue,urlcolor=blue}

\usepackage{graphicx}

\usepackage[margin=1in]{geometry}
\newtheorem{lemma}{Lemma}
\newtheorem{theorem}{Theorem}
\newtheorem{corollary}{Corollary}

\theoremstyle{definition}

\newcommand{\val}{\mathsf{val}}

\DeclareMathOperator*{\bE}{\mathbb{E}}
\newcommand{\bR}{\mathbb{R}}
\newcommand{\poa}{\mathsf{PoA}}
\newcommand{\eps}{\varepsilon}
\newcommand{\argmax}{\operatorname*{argmax}}
\newcommand{\rw}{\mathsf{rw}}
\newcommand{\dt}{\,\mathrm{d}t}
\renewcommand{\v}{\boldsymbol{v}}
\renewcommand{\b}{\boldsymbol{b}}

\title{Efficiency of the First-Price Auction in the Autobidding World}
\author{
    Yuan Deng\thanks{Google Research, \texttt{dengyuan@google.com}}
    \and
    Jieming Mao\thanks{Google Research, \texttt{maojm@google.com}}
    \and
    Vahab Mirrokni\thanks{Google Research, \texttt{mirrokni@google.com}}
    \and
    Hanrui Zhang\thanks{Carnegie Mellon University, \texttt{hanruiz1@cs.cmu.edu}}
    \and
    Song Zuo\thanks{Google Research, \texttt{szuo@google.com}}
}
\date{}

\begin{document}

\maketitle
\thispagestyle{empty}
\begin{abstract}
    We study the price of anarchy of the first-price auction in the autobidding world, where bidders can be either utility maximizers (i.e., traditional bidders) or value maximizers (i.e., autobidders).  We show that with autobidders only, the price of anarchy of the first-price auction is $1/2$, and with both kinds of bidders, the price of anarchy degrades to about $0.457$ (the precise number is given by an optimization).  These results complement the recent result by \citet{JinL22} showing that the price of anarchy of the first-price auction with traditional bidders only is $1 - 1/e^2$.  We further investigate a setting where the seller can utilize machine-learned advice to improve the efficiency of the auctions.  There, we show that as the accuracy of the advice increases, the price of anarchy improves smoothly from about $0.457$ to $1$.
\end{abstract}


\section{Introduction}

Autobidding, the procedure of delegating the bidding tasks to automated agents to procure advertisement opportunities in online ad auctions, is becoming the prevalent bidding methods and contributing to more than 80\% of total online advertising traffic \citep{WinNT}. Autobidding significantly simplifies the interaction between the advertisers and the ad platform: instead of submitting bids for each auction separately, the advertisers only need to specify their high-level objectives and constraints and the automated agents bid on behalf of the advertisers. A popular autobidding strategy is value maximization subject to a target return-on-investment (ROI) constraint, in which the 
autobidding agents aim to maximize their value, such as the number of clicks/conversions, subject to a minimum admissible ratio on the return-on-investment constraint. This type of bidders has been referred as {\em value maximizers} by the recent line of research on autobidding~\citep{deng2021towards,balseiro2021landscape}. 

As the behavior model of automated bidders is very different from the classic (quasi-linear) {\em utility maximizers} maximizing the difference between value and payment, there has been a growing body of literature revisiting the effectiveness of the existing mechanisms through the lens of the autobidding world~\citep{aggarwal2019autobidding,deng2021towards,balseiro2021landscape,rFPApaper,babaioff2020non}. In particular, it is well-known that running the second-price auction against utility maximizers results in socially optimal outcomes; however, \citet{aggarwal2019autobidding} show that the price of anarchy (PoA) \citep{koutsoupias1999worst}, the ratio between the worst welfare in equilibrium and the socially optimal welfare, of running the second-price auction against value maximizers is $1/2$ (when the bidders do not adopt dominated strategies\footnote{When discussing the second-price auction, we assume the bidders do not adopt dominated strategies. In the case when the bidders could adopt dominated strategies, the PoA of running the second-price auction can be $0$.}).

In addition to the change in bidding methods, there has been an industry-wide change in auction formats in recent years: shifting from the second-price auction to the first-price auction, especially in the display ad markets, culminating
in the move to the first-price auction from the Google Ad Exchange in September 2019 \citep{paes2020competitive}. The shift towards the first-price auction in combination with the increasing adoption of autobidding raises the question of understanding {\em the PoA of running the first-price auction in the autobidding world}.

\subsection{Our Results}

In this paper, we charaterize the PoA of running the first-price auction in the fully autobidding world, where all bidders are value maximizers, and in the mixed autobidding world, where bidders could be either value maximizers or utility maximizers. Our results complement the very recent result from \cite{JinL22} showing that the PoA of running the first-price auction is $1 - 1 / e^2$ in a world without autobidding, i.e., all bidders are utiltiy maximizers. The results are summarized in Table~\ref{tab:main_poa_result} and we also include the known results for second-price auctions for comparison.

\begin{table}[h]
    \centering
    \begin{tabular}{|c||c|c|}
       \hline
       Mechanism  & First-price auction & Second-price auction \\
       \hline
       \hline
       Full autobidding  & $1/2$ & $1/2$~\citep{aggarwal2019autobidding} \\
       \hline
       Mixed autobidding & $\min_{t \in [0, 1]} \frac{1 + t \ln t}{2 - t + t \ln t} \approx 0.457$ & $1/2$~\citep{BalseiroDMMZ21} \\
       \hline
       No autobidding & $1 - 1 / e^2 \approx 0.865$~\citep{JinL22} & $1$~\citep{vickrey1961counterspeculation} \\
       \hline
    \end{tabular}
    \caption{Summary of PoAs in second-price and first-price auctions under full autobidding, mixed autobidding, and no autobidding environments. We note that for the second-price auction, bidders are assumed to not adopt dominated strategies; otherwise the PoA can be $0$.}
    \label{tab:main_poa_result}
\end{table}

Our results demonstrate that in a full autobidding world, the PoA of running the second-price auction and the first-price auction is the same, i.e, $1/2$; however, in the mixed autobidding world with both utility maximizers and value maximizers, the PoA of running the first-price auction is strictly less than the PoA of running the second-price auction. It is worth highlighting that our PoA bound for the first-price auction in the mixed autobidding world involves solving an optimization problem $\min_{t \in [0, 1]} \frac{1 + t \ln t}{2 - t + t \ln t}$ and it turns out that this is the right bound as we are able to construct an instance matching this bound.

In addition to the above results, we further adopt the environment introduced by \citet{BalseiroDMMZ21}, where the seller can leverage machine-learned advice that approximates the buyers' values for designing mechanisms. \citet{BalseiroDMMZ21} show that as the advice gets more and more accurate in predicting the buyers' values, setting reserves based on the advice in the second-price auction makes the PoA approaching $1$. In our setting of using machine-learned advice in the first-price auction, we prove that setting reserves based on machine-learned advice is also effective in improving PoA in both full and mixed autobidding world, and the PoA approaches $1$ as the quality of the machine-learned advice increases. 

\subsection{Related Works}
For first-price auction with utility maximizers, the price of anarchy has been extensively studied. \citet{syrgkanis2013composable} first showed that the PoA is at least $1 - 1 / e$, which was later improved to $\approx 0.743$ by \citet{hoy2018tighter}. On the other hand, \citet{hartline2014price} gave a concrete instance to establish an upper bound of $\approx 0.869$. In a very recent paper, \citet{JinL22} show that the price of anarchy is exactly $1 - 1/e^2$. \citet{paes2020competitive} provide models and analyses to explain why Display Ads market shifts to use first-price auctions.

Motivated by the emergence of autobidding in online advertising, there is a recent line of work studying auction design for value maximizers in an autobidding world. \citet{aggarwal2019autobidding} design optimal bidding strategies in truthful auctions, show the existence of equilibrium, and prove welfare price of anarchy results. \citet{deng2021towards} and \citet{BalseiroDMMZ21} show how boosts and reserves can be used to improve the welfare efficiency guarantees given machine-learned advice approximating buyers' values. \citet{balseiro2021landscape,BalseiroDMMZ22} give the characterization of the revenue-optimal auctions under various information structure. \cite{Mehta22} and  \cite{rFPApaper} study how randomized and non-truthful mechanisms can help improve welfare efficiency for the case of two bidders. In particular, \cite{rFPApaper} show that the PoA of running the first-price auction is $1/2$ when restricting to equilibria of deterministic bids. In this work, we consider general equilibria with randomized bidding strategies.

\section{Preliminaries}

\paragraph{Ad auctions.}
Following prior work on autobidding, we consider a setting with $n$ bidders $[n]$ and $m$ auctions $[m]$, where each bidder participates in all auctions simultaneously. We generally use $i$ to index bidders, and $j$ to index auctions.
Each bidder $i$ has a value $v_{i, j} \geq 0$ for winning in each auction $j$.
We consider first-price auctions --- arguably the most common and well-known auction format --- as the auction mechanism throughout the paper.
That is, each bidder $i$ submits a bid $b_{i, j}$ in each auction $j$.
In each auction $j$, the bidder with the highest bid wins and is charged the winning bid as the payment.
For the existence of equilibrium in the first-price auction, we assume ties are broken in favor of the bidder with the higher value within an auction, which is a standard assumption in the literature \citep{maskin2000equilibrium}.\footnote{\citet{maskin2000equilibrium} assume ties are broken with a second-round second-price auction among the highest bidders to guarantee the existence of equilibrium in the presence of discrete values. This tie-breaking rule is equivalent to favoring the bidder with the higher value for standard utility maximizers, while the two tie-breaking rules can be different for value maximizers. We choose the one without the additional definition of second-round second-price auction for simplicity. 
}
If multiple bidders have the same value, then ties are broken in favor of the bidder with the smallest index.\footnote{Our positive results are oblivious to tiebreaking, i.e., they remain valid no matter what tiebreaking rule is used.}
We let $x_{i, j} = x_{i, j}(\{b_{i', j}\}_{i'})$ be the indicator variable indicating that $i$ wins in auction $j$.
Then, the value each bidder $i$ receives in auction $j$ is $\val_{i, j} = \val_{i, j}(\{b_{i', j}\}_{i'}) = x_{i, j}(\{b_{i', j}\}_{i'}) \cdot v_{i, j}$, and the payment $i$ makes in auction $j$ is $p_{i, j} = p_{i, j}(\{b_{i', j}\}_{i'}) = x_{i, j}(\{b_{i', j}\}_{i'}) \cdot b_{i, j}$.
We remark that bidders may employ randomized bidding strategies, in which case $\{b_{i, j}\}$, $\{x_{i, j}\}$, $\{\val_{i, j}\}$ and $\{p_{i, j}\}$ are all random variables.
Throughout the paper, we omit the dependence of $x_{i, j}$, $\val_{i, j}$ and $p_{i, j}$ on $\{b_{i', j}\}_{i'}$ whenever it is clear from the context.
For brevity, we write $\v_i = \{v_{i, j}\}_j$, $\b_i = \{b_{i, j}\}_j$, etc. 

\paragraph{Utility maximizers and value maximizers.}
We consider two types of bidders: traditional (quasi-linear) utility maximizers and (ROI-constrained) value maximizers.
In the rest of the paper, we use $B_u$ to denote the set of utility maximizers and $B_v$ the set of value maximizers, where we always have $B_u \cap B_v = \emptyset$ and $B_u \cup B_v = [n]$.
For each bidder $i$, let $\val_i = \sum_j \val_{i, j}$ be the total value $i$ receives from all auctions, and $p_i = \sum_j p_{i, j}$ be the total payment $i$ makes.
Fixing other bidders' bidding strategies, a utility maximizer $i$ bids in a way such that the total quasi-linear utility, i.e., $\val_i - p_i$, is maximized, without additional constraints.
On the other hand, a value maximizer $i$ bids in a way such that the total value $\val_i$ that $i$ receives is maximized, subject to the ROI constraint that the ratio between the total value $\val_i$ and the total payment $p_i$ is at most some predetermined threshold. Without loss of generality we assume this threshold is $1$, i.e., we require the total payment $p_i$ is at most the total value $\val_i$ for each value maximizer $i$.
Formally, fixing other bidders' bids, a value maximizer $i$ solves the following optimization problem to determine their bids $\{b_{i, j}\}_j$:
\begin{align*}
    \text{maximize} \qquad & \bE\left[\val_i\right]  \\
    \text{subject to} \qquad & \bE\left[\val_i - p_i\right] \ge 0.
\end{align*}
Here, the expectations are taken over the randomness introduced by bidders bidding randomly.
We note that when deciding how to bid, $i$ can observe the distributions, but not the realizations, of the other bidders' bids.

\paragraph{Equilibria and the price of anarchy.}
The focus of this paper is on the efficiency of the first-price auction in equilibrium.
That is, when each bidder $i$'s bidding strategy $\b_i$ is optimal for that bidder's objective, given all other bidders' bidding strategies $\{\b_{i'}\}_{i' \ne i}$ (or $\b_{-i}$ for brevity).
In other words, each bidder's strategy is a best response to other bidders' strategies.
Note that since each bidder $i$ cannot observe the realizations of other bidders' bids, all bidders' bids $\{\b_i\}_i$ must be independent.
The way we measure efficiency is through the classical notion of the price of anarchy (PoA) \citep{koutsoupias1999worst}, i.e., the ratio between the welfare resulting from the worst equilibrium possible, and the optimal welfare.
Formally, given a problem instance of ($n$, $m$, $B_u$, $B_v$, $\{v_{i, j}\}$), we are interested in the instance-wise PoA defined as follows:
\[
    \poa(n, m, B_u, B_v, \{v_{i, j}\}) = \inf_{\{\b_i\}_i~\text{form an equilibrium}} \frac{\sum_i \bE[\val_i]}{\sum_j \max_i v_{i, j}}.
\]
Our goal is to pin down the PoA of the first-price auction, which is the infimum of the instance-wise PoA taken over all instances of the problem.
We assume $\sum_j \max_i v_{i, j} > 0$ in order for the PoA to be well-defined.

\section{PoA of the First-price Auction in a Full Autobidding World}

In the literature~\citep{aggarwal2019autobidding,deng2021towards,BalseiroDMMZ21}, the analyses of PoA for the second-price auction in a full autobididng world mainly leverage the fact that the second-price auction is truthful for utility maximizers, a nice property implying that value maximizers adopt {\em uniform bidding} strategy in equilibrium~\citep{aggarwal2019autobidding}. Here, a uniform bidding strategy for a value maximizer $i$ is a strategy in which she bids $k_i \cdot v_{i,j}$ with a constant bid multiplier $k_i$ in each auction $j$. However, as the first-price auction is not truthful for utility maximizers, a uniform bidding strategy is generally no longer a best response for a value maximizer. A value maximizer $i$ may adopt a {\em non-uniform} bidding strategy using different bid multiplier $k_{i,j}$ in different auction $j$, and the bid multipliers can even be {\em randomized}.

\cite{rFPApaper} show that the PoA of running the first-price auction in a full autobidding world is $1/2$ when restricting to equilibria of bidding in a determistic way.
In this section, we devise a novel argument to analyze equilibria with possibly randomized non-uniform bidding strategies to show that the PoA of the first-price auction in a full autobidding world is exactly $1/2$.

\begin{theorem}
\label{thm:value}
    When all bidders are value maximizers, the PoA of the first-price auction in a full autobidding world is $1/2$.
    Formally,
    \[
        \inf_{n, m, \{v_{i, j}\}} \poa(n, m, B_u = \emptyset, B_v = [n], \{v_{i, j}\}) = 1/2.
    \]
\end{theorem}
\begin{proof}
    First we provide a simple hard instance, which implies an upper bound of $1/2$ on the PoA.
    Consider an instance with $n = 2$ bidders and $m = 2$ auctions, where both bidders are value maximizers.
    We set $v_{1, 1} = 1$, $v_{2, 2} = 1 - \eps$, and $v_{1, 2} = v_{2, 1} = 0$, where $\eps > 0$ is a small quantity to be fixed later.
    Consider the following deterministic bidding strategies: $b_{1, 1} = b_{2, 1} = 0$ (so bidder $1$ wins in auction $1$ by tiebreaking and pays $0$), $b_{1, 2} = 1$, and $b_{2, 2} = 0$ (so bidder $1$ wins in auction $2$ and pays $1$).
    First we verify this is in fact an equilibrium.
    Observe that the ROI constraints of both bidders are satisfied.
    Bidder $1$ does not want to deviate from $\b_1$, since bidder $1$ already gets the maximum value possible, i.e., $1$.
    Bidder $2$ would ideally want to bid higher (possibly randomly) in auction $2$ in order to win with some probability and get positive value.
    However, given $b_{1, 2} = 1$, bidder $2$ would have to bid at least $1$ in order to win, which means $p_{2, 2} \ge 1$ whenever $x_{2, 2} = 1$.
    This would violate bidder $2$'s ROI constraint.
    So both bidders' strategies are best responses, and therefore form an equilibrium.
    On the other hand, the optimal welfare is $v_{1, 1} + v_{2, 2} = 2 - \eps$, and the welfare resulting from the above equilibrium is only $v_{1, 1} = 1$, which means the PoA is at most $1/(2 - \eps)$.
    Letting $\eps \to 0$ yields an upper bound of $1/2$.
    \medskip

    Now we focus on lower bounding the instance-wise PoA fixing any $n$, $m$, and $\{v_{i, j}\}$.
    Consider any equilibrium, induced by bidding strategies $\{\b_i\}_i$.
    Our task is to show that
    \[
        2 \sum_i \bE[\val_i] \ge \sum_j \max_i v_{i, j},
    \]
    where the expectations, as well as all expectations in the rest of the proof, are taken over $\{\b_i\}_i$.
    For each auction $j$, let $\rw(j) = \argmax_i v_{i, j}$ be the ``rightful winner'' of auction $j$, where ties are broken in favor of the bidder with the smaller index.
    Now fix a bidder $i$, and consider the following proxy bidding strategy $\b'_i$ of bidder $i$: $b'_{i, j} = v_{i, j}$ in each auction $j$, i.e., the proxy bidding strategy $\b'_i$ is a strategy in which bidder $i$ bids truthfully.
    Then, for each auction $j$ where $i = \rw(j)$, the probability that $i$ wins in auction $j$ under the proxy bidding strategy $\b'_{i}$, fixing all other bidders' strategies $\b_{-i}$, is
    \[
        \Pr_{\b_{-i}}[x_{i, j}(b'_{i,j}, \b_{-i,j}) = 1] = \Pr_{\b_{-i}}\left[\max_{i' \ne i} b_{i',j} \le b'_{i, j}\right] = \Pr_{\b_{-i}}\left[\max_{i' \ne i} b_{i', j} \le v_{i, j}\right].
    \]
    Also, observe that $i$'s ROI constraint is satisfied under the proxy bidding strategy $\b'_i$, since under the proxy bidding strategy $\b'_i$, the following holds deterministically:
    \[
        \val_i(\b'_{i}, \b_{-i}) - p_i(\b'_{i}, \b_{-i}) = \sum_j x_{i, j}(b'_{i,j}, \b_{-i,j}) \cdot (v_{i, j} - b'_{i, j}) = \sum_j x_{i, j}(b'_{i,j}, \b_{-i,j}) \cdot (v_{i, j} - v_{i, j}) = 0.
    \]
    Taking the expectation over $\b_{-i}$, we see that $i$'s ROI constraint is satisfied (and is in fact binding).
    The above observations imply the following fact: There exists a bidding strategy (i.e., $\b'_i$), under which $i$'s ROI constraint is satisfied, and the expected value $i$ receives is
    \[
        \bE[\val_i(\b'_{i}, \b_{-i})] \ge \sum_{j: i = \rw(j)} \Pr_{\b_{-i}}[x_{i, j}(b'_{i,j}, \b_{-i,j}) = 1] \cdot v_{i, j} = \sum_{j: i = \rw(j)} \Pr_{\b_{-i}}\left[\max_{i' \ne i} b_{i', j} \le v_{i, j}\right] \cdot v_{i, j}.
    \]
    Observe that since $i$ is a value maximizer and $\b_i$ is a best response to $\b_{-i}$, we must have
    \[
        \bE[\val_i(\b_{i}, \b_{-i})] \ge \bE[\val_i(\b'_{i}, \b_{-i})] \ge \sum_{j: i = \rw(j)} \Pr_{\b_{-i}}\left[\max_{i' \ne i} b_{i', j} \le v_{i, j}\right] \cdot v_{i, j}.
    \]
    On the other hand, fix a bidder $i$ and consider the total payment collected in any auction $j$ with $\rw(j) = i$, under $(\b_i, \b_{-i})$.
    By the properties of the first-price auction, we have
    \begin{align*}
        \sum_{i'} \bE[p_{i', j}(b_{i,j}, \b_{-i,j})] & = \bE\left[\max_{i'} b_{i', j}\right] \ge \bE\left[\max_{i' \ne i} b_{i', j}\right] \\
        & \ge \Pr_{\b_{-i}}\left[\max_{i' \ne i} b_{i', j} > v_{i, j}\right] \cdot v_{i, j} = \left(1 - \Pr_{\b_{-i}}\left[\max_{i' \ne i} b_{i', j} \le v_{i, j}\right]\right) \cdot v_{i, j}.
    \end{align*}
    Summing over all auctions $j$ with $\rw(j) = i$, we have
    \[
        \sum_{i', j: i = \rw(j)} \bE[p_{i', j}(\b_{i}, \b_{-i})] \ge \sum_{j: i = \rw(j)} \left(1 - \Pr_{\b_{-i}}\left[\max_{i' \ne i} b_{i', j} \le v_{i, j}\right]\right) \cdot v_{i, j}.
    \]
    Combining this with the bound on the total expected value $i$ receives, we get
    \[
       \bE[\val_i(\b_{i}, \b_{-i})] + \sum_{i', j: i = \rw(j)} \bE[p_{i', j}(\b_{i}, \b_{-i})] \ge \sum_{j: i = \rw(j)} v_{i, j}.
    \]
    Further summing over $i$ (and omitting dependence on $(\b_{i}, \b_{-i})$), the above gives
    \[
        \sum_i \bE[\val_i] + \sum_i \bE[p_i] \ge \sum_j \max_i v_{i, j}.
    \]
    Finally, observe that for any bidder $i$, the ROI constraint implies $\bE[p_i] \le \bE[\val_i]$, and
    therefore we can conclude the proof with
    \[
        2 \sum_i \bE[\val_i] \ge \sum_i \bE[\val_i] + \sum_i \bE[p_i] \ge \sum_j \max_i v_{i, j}. \qedhere
    \]
\end{proof}

We remark that the proof of Theorem~\ref{thm:value} in fact establishes the following claim, which will be a building block in our analysis of PoA of the first-price auction in a mixed autobidding world with both value maximizers and utility maximizers.

\begin{lemma}
\label{lem:value}
    For any $n$, $m$, $B_u$, $B_v$, and $\{v_{i, j}\}$, if $\{\b_i\}_i$ form an equilibrium, then
    \[
        \sum_{i \in B_v} \bE[\val_i] + \sum_{i, j: \rw(j) \in B_v} \bE[p_{i, j}] \ge \sum_{j: \rw(j) \in B_v} \max_i v_{i, j}.
    \]
\end{lemma}

We also remark that it is not hard to generalize Theorem~\ref{thm:value} to the multi-slot setting, where multiple slots are sold in each auction, i.e., the generalized first-price auction.
In the multi-slot setting, the click-through rates of the slots are different, and the value that a bidder $i$ receives when winning a slot in an auction $j$ is the product of $v_{i, j}$ and the click-through rate of that particular slot. One can show that the PoA of the generalized first-price auction in a full autobiding world is also $1/2$, by adapting the proof of Theorem~\ref{thm:value}.

\section{PoA of the First-price Auction in a Mixed Autobidding World}

In this section, we proceed to the main result of the paper: characterizing the exact PoA of running the first-price auction in a mixed autobidding world with both utility maximizers and value maximizers participating in the auctions.
The entire section is devoted to the proof of the following claim.

\begin{theorem}
\label{thm:main}
    The PoA of the first-price auction in a mixed autobidding world is given by:
    \[
        \inf_{n, m, B_u, B_v, \{v_{i, j}\}} \poa(n, m, B_u, B_v, \{v_{i, j}\}) = \min_{t \in [0, 1]} \frac{1 + t \ln t}{2 - t + t \ln t} \approx 0.457.
    \]
\end{theorem}

\paragraph{Roadmap.}
In the rest of the section, we first establish a PoA lower bound.
At first sight, one might be tempted to adapt the proof of Theorem~\ref{thm:value}, hoping to get the same lower bound of $1/2$ in the presence of utility maximizers.
In doing so, one faces the following intuitive difficulty: for a utility maximizer $i$, the value $i$ receives when best responding is not necessarily lower bounded by the value $i$ receives when using the proxy bidding strategy $b'_{i, j} = v_{i, j}$.
This is because aiming to maximize utility, $i$ may bid in a way such that both the value $i$ receives and the payment $i$ makes are smaller, but their difference (which is the utility $i$ gets) is larger.
One simple fix is to consider instead the following proxy strategy: $b''_{i, j} = v_{i, j} / 2$.
Adapting the proof of Theorem~\ref{thm:value}, it is not hard to show this alternative proxy strategy gives a PoA lower bound of $1/4$.
However, it is unlikely that such a fix would yield the tight ratio of about $0.457$.

In order to overcome this obstacle, we present a novel local analysis (Lemma~\ref{lem:local}) for the first-price auction that handles auctions where the rightful winner is a utility maximizer, which roughly says that if the expected value that the rightful winner receives from such an auction is small, then the expected payment that other bidders make in the same auction must be large.
This argument relies on a characterization of the structure of the best response of a utility maximizer.
Given the local analysis, one can already get a lower bound of about $0.317$ on the PoA by exploiting the tradeoff between value and payment.

To obtain the right ratio, we further combine the local analysis for utility maximizers with the bound for value maximizers (Lemma~\ref{lem:value}) discussed earlier.
The idea is to consider $4$ quantities separately: ($A$) the total value that all value maximizers receive, ($B$) the total value that all utility maximizers receive from auctions where they are the rightful winner, ($C$) the total payment made in auctions where the rightful winner is a value maximizer, and ($D$) the total payment made in auctions where the rightful winner is a utility maximizer, by bidders who are not the rightful winner.

One can show that the welfare in equilibrium is lower bounded by a very specific combination of these $4$ quantities (Lemma~\ref{lem:combination}).
Moreover, Lemma~\ref{lem:value} provides a tradeoff between ($A$) and ($C$), and the local analysis (Lemma~\ref{lem:local}) provides a tradeoff between ($B$) and ($D$).
Therefore, one can lower bound the PoA by solving an optimization problem involving a few number of parameters to get the worst-case tradeoffs. We show that this optimization gives the right ratio of about $0.457$ (Lemma~\ref{lem:ub}). Finally, we present a problem instance where there exists an equilibrium where the gap between the actual welfare and the optimal welfare is precisely the right ratio (Lemma~\ref{lem:lb}).
This, together with the matching lower bound, concludes the proof of Theorem~\ref{thm:main}.

\subsection{A Local Analysis}

In this subsection, we propose a local analysis of the first-price auction to analyze the tradeoff between welfare and payment in equillibrium when the rightful winner of the auction is a utility maximizer. Although the analysis is mainly designed for utility maximizers towards establishing the lower bound of PoA, it also works when the rightful winner of the auction is a value maximizer. For the remainder of this subsection, we fix $n$, $m$, $\{v_{i, j}\}$, and an equilibrium profile $\{\b_i\}_i$, and consider any fixed bidder $i$ and auction $j$ where $i$ is the rightful winner. For brevity, we also let $v = v_{i, j}$.

\paragraph{Value-payment frontiers.}
We first try to understand the maximum expected value $i$ can get in auction $j$ when paying a particular amount in expectation.
Let $F$ be the CDF of the highest other bid, i.e., for any $x \ge 0$,
\[
    F(x) = \Pr_{\b_{-i}}\left[\max_{i' \ne i} b_{i', j} \le x\right].
\]
Note that in the first-price auction $j$, if bidder $i$ bids $t$, the bidder $i$ pays $t \cdot F(t)$ and earns value $v \cdot F(t)$. For any $x \ge 0$, consider a function $G: \bR_+ \to [0, v]$ such that
\[
    G(x) = \sup_{t \ge 0:~t \cdot F(t) \le x} v \cdot F(t).
\]
In other words, $G$ maps a desired expected payment $x$ to the maximum (or strictly speaking, supremum) value $i$ can get by bidding deterministically, when paying at most the desired payment $x$ in expectation.
By bidding randomly, $i$ can further achieve any value-payment pair that is a convex combination of points on $G$.
In light of this, we further consider the {\em concave envelope $H$ of $G$}, which captures the value-payment frontier when $i$ can bid randomly.
Note that for utility maximizers, although there always exists a deterministic best response (since all bids with a positive probability in a randomized bidding strategy must result in the same quasi-linear utility), the equilibrium bidding strategies may still be randomized for the sake of stability.

\paragraph{Characterizing the best response.}
Given the definition of the value-payment frontier $H$, one can further observe that when best responding, $i$'s expected utility $\bE[\val_{i, j} - p_{i, j}]$ depends only on $H$.
In fact, intuitively, since $H$ is concave, $i$'s expected value and payment in auction $j$ must correspond to the ``rightmost'' point on $H$ where the derivative is larger than or equal to $1$ (this is informal since (1) the point may not be unique and (2) the derivative may be undefined).
Formally, when $i$ is a utility maximizer, $i$'s best response in $j$, $b_{i, j}$, has the following property.

\begin{lemma}
\label{lem:marginal}
    For any $b \ge 0$ in the support of $b_{i, j}$, we have for any $w \ge b \cdot F(b)$, $H(w) - H(b \cdot F(b)) \le w - b \cdot F(b)$.
\end{lemma}
\begin{proof}
    Suppose towards a contradiction that there exists $w > b \cdot F(b)$ such that $H(w) - H(b \cdot F(b)) > w - b \cdot F(b)$, or equivalently, $H(w) - w > H(b \cdot F(b)) - b \cdot F(b)$.
    Let $b'_{i, j}$ be a possibly randomized bidding strategy that realizes the point $(w, H(w))$, i.e., when using $b'_{i, j}$, $i$ receives value $H(w)$ and pays $w$ in expectation in auction $j$.
    Then, the utility that $i$ receives from $j$ when using the new strategy $b'_{i, j}$ is $H(w) - w$.
    On the other hand, consider the deterministic bidding strategy $b''_{i, j} = b$.
    Since $b$ is in the support of $b_{i, j}$, $b''_{i, j}$ must also be a best response.
    However, under $b''_{i, j}$, $i$ receives utility $H(b \cdot F(b)) - b \cdot F(b)$, which is strictly worse than the utility $i$ receives when using the alternative strategy $b'_{i, j}$. In other words, $b''_{i, j}$ cannot be a best response, a contradiction.
\end{proof}
We remark that the same property also holds when $i$ is a value maximizer, with minor modifications in the proof.

\paragraph{A tradeoff between value and payment in equilibrium.}
We are ready to present a tradeoff between the value $i$ receives and the payment other bidders make in auction $j$.
The intuition is that given Lemma~\ref{lem:marginal}, the curve $H$ must be relatively flat to the right of $(p_{i,j}, H(p_{i,j}))$ (i.e., with slope at most $1$), which means the CDF function $F$ of the highest other bid must also be with a relatively small slope. On the other hand, the payment made by other bidders can be written as
\[
    b_{i, j} \cdot (1 - F(b_{i, j})) + \int_{b_{i, j}}^\infty (1 - F(t)) \dt.
\]
This means if $F(b_{i, j})$, the probability that $i$ wins when bidding $b_{i, j}$, is small, then since $F$ is with a small slope to the right of $b_{i, j}$, the payment must be relatively large.
The following claim formally captures this tradeoff.

\begin{lemma}
\label{lem:local}
    There exists some $x \in [0, 1]$ such that
    \[
        \bE[\val_{i, j}] = x \cdot v \quad \text{and} \quad \sum_{i' \ne i} \bE[p_{i', j}] \ge (1 - x + x \ln x) \cdot v.
    \]
\end{lemma}
\begin{proof}
    Let $x = \bE[F(b_{i, j})]$.
    Clearly $x \in [0, 1]$ and $\bE[\val_{i, j}] = x \cdot v$.
    Our remaining task is to show
    \[
        \sum_{i' \ne i} \bE[p_{i', j}] \ge (1 - x + x \ln x) \cdot v.
    \]
    Recall that
    \[
        \sum_{i' \ne i} \bE[p_{i', j}] = \bE_{b_{i, j}} \left[b_{i, j} \cdot (1 - F(b_{i, j})) + \int_{b_{i, j}}^\infty (1 - F(t)) \dt\right].
    \]
    As a result, in order to lower bound the payment, we want to pointwise upper bound $F$. Fix any $b$ in the support of the best-response strategy $b_{i,j}$, we must have 
    \[
      H(b \cdot F(b)) = G(b \cdot F(b)) = v \cdot F(b).
    \]
    By Lemma~\ref{lem:marginal}, for any $w \ge b \cdot F(b)$, we have $H(w) - H(b \cdot F(b)) \le w - b \cdot F(b)$, which implies that $H(w) \le H(b \cdot F(b)) + w - b \cdot F(b)$.
    Also, since $H$ is the concave envelope of $G$, the same upper bound applies to $G$, and we have
    \[
        G(w) \le H(w) \leq H(b \cdot F(b)) + w - b \cdot F(b) = v \cdot F(b) + w - b \cdot F(b).
    \]
    Due to the monotonicity of $F(x)$, we have that for any $t \ge b$, we have $t \cdot F(t) \ge b \cdot F(b)$. Therefore, given the relation between $G$ and $F$, the above implies that for any $t \in [b, v)$ (note that this range might be empty),
    \[
        v \cdot F(t) \leq G(t \cdot F(t)) \le v \cdot F(b) + t \cdot F(t) - b \cdot F(b) \implies F(t) \le \frac{(v - b) \cdot F(b)}{v - t}.
    \]
    We next bound the payment conditioned on $b_{i, j} = b$.
    We consider two possible cases: $b \ge v$ and $b < v$.
    If $b \ge v$, we simply have
    \begin{align*}
        \sum_{i' \ne i} \bE[p_{i', j} \mid b_{i, j} = b] & = b \cdot (1 - F(b)) + \int_{b}^\infty (1 - F(t)) \dt \\
        & \ge b \cdot (1 - F(b)) \ge (1 - F(b)) \cdot v \ge (1 - F(b) + F(b) \ln F(b)) \cdot v.
    \end{align*}
    In particular, the last inequality is because $F(b) \ln F(b) \le 0$ when $F(b) \in [0, 1]$.
    If $b < v$, we must have $b = v - (v - b) \cdot 1 \le v - (v - b) \cdot F(b)$, and as a result,
    \begin{align*}
        \sum_{i' \ne i} \bE[p_{i', j} \mid b_{i, j} = b] & = b \cdot (1 - F(b)) + \int_{b}^\infty (1 - F(t)) \dt \\
        & \ge b \cdot (1 - F(b)) + \int_{b}^{v - (v - b) \cdot F(b)} \left(1 - \frac{(v - b) \cdot F(b)}{v - t}\right) \dt \\
        & = (1 - F(b)) \cdot v + (v - b) \cdot F(b) \cdot \ln (v - t)\Big|_b^{v - (v - b) \cdot F(b)} \\
        & = (1 - F(b)) \cdot v + (v - b) \cdot F(b) \cdot \ln F(b)  \\
        & \ge (1 - F(b) + F(b) \ln F(b)) \cdot v.
    \end{align*}
    Combining the two cases and further taking the expectation over $b_{i, j}$, we get
    \[
        \sum_{i' \ne i} \bE[p_{i', j}] \ge \bE_{b_{i,j}}[(1 - F(b_{i, j}) + F(b_{i, j}) \ln F(b_{i, j})) \cdot v] \ge (1 - x + x \ln x) \cdot v,
    \]
    where the last inequality is due to Jensen's inequality, the fact that $z \mapsto 1 - z + z \ln z$ is convex on $[0, 1]$, and we chose $x = \bE[F(b_{i, j})]$.
\end{proof}

We remark that Lemma~\ref{lem:local} alone already implies a lower bound of about $0.317$ on the PoA: summing the bound in Lemma~\ref{lem:local} over $i$ and $j$ where $\rw(j) = i$ and exploiting the convexity of the mapping $z \mapsto 1 - z + z \ln z$, one can show there is some $x \in [0, 1]$ such that
\[
    \sum_i \bE[\val_i] \ge x \cdot \sum_j \max_i v_{i, j} \quad \text{and} \quad \sum_i \bE[p_i] \ge (1 - x + x \ln x) \cdot \sum_j \max_i v_{i, j}.
\]
Since $\sum_i \bE[\val_i] \ge \sum_i \bE[p_i]$, the worst case scenario is when $x = 1 - x + x \ln x$, which gives $x \approx 0.317$, a lower bound on the PoA.

\subsection{Putting Everything Together}

In this subsection, we combine Lemma~\ref{lem:value} and Lemma~\ref{lem:local} to obtain an improved (and in fact, tight) lower bound on the PoA of the first-price auction in the mixed autobidding world.
Fixing $n$, $m$, $B_u$, $B_v$ and $\{v_{i, j}\}$, we consider four key quantities in the proof. ($A$): the total value that all value maximizers receive; ($B$): the total value that all utility maximizers receive from auctions where they are the rightful winner; ($C$): the total payment made in auctions where the rightful winner is a value maximizer; and ($D$): the total payment made in auctions where the rightful winner is a utility maximizer, by bidders who are not the rightful winner.
Formally,
\begin{align*}
    & A = \sum_{i \in B_v} \bE[\val_i], && B = \sum_{i \in B_u, j: \rw(j) = i} \bE[\val_{i, j}] \\
    & C = \sum_{i, j: \rw(j) \in B_v} \bE[p_{i, j}], && D = \sum_{i, j: \rw(j) \in B_u \text{ and } \rw(j) \ne i} \bE[p_{i, j}].
\end{align*}
We first show that the welfare can be bounded by a combination of the above quantities.

\begin{lemma}
\label{lem:combination}
    For any bidding strategies $\{\b_i\}_i$ that form an equilibrium, we always have
    \[
        \sum_i \bE[\val_i] \ge \max\{A, C + D\} + B.
    \]
\end{lemma}
\begin{proof}
    Observe that
    \begin{align*}
        \sum_i \bE[\val_i] & = \sum_{i \in B_v} \bE[\val_i] + \sum_{i \in B_u} \bE[\val_i] \\
        & = \max\left\{\sum_{i \in B_v} \bE[\val_i], \sum_{i \in B_v} \bE[p_i]\right\} + \sum_{i \in B_u} \bE[\val_i] \\
        & = \max\left\{\sum_{i \in B_v} \bE[\val_i], \sum_{i \in B_v} \bE[p_i]\right\} + \sum_{i \in B_u, j: \rw(j) = i} \bE[\val_{i, j}] + \sum_{i \in B_u, j: \rw(j) \ne i} \bE[\val_{i, j}] \\
        & = \max\left\{A, \sum_{i \in B_v} \bE[p_i]\right\} + B + \sum_{i \in B_u, j: \rw(j) \ne i} \bE[\val_{i, j}].
    \end{align*}
    Note that for each $i \in B_u$ and $j$, we always have $\bE[\val_{i, j}] \ge \bE[p_{i, j}]$ since $i$ is best responding and the utility $i$ receives in $j$ is at least $0$. Therefore, we can further relax the above into the following:
    \begin{align*}
        \sum_i \bE[\val_i] & \ge \max\left\{A, \sum_{i \in B_v} \bE[p_i]\right\} + B + \sum_{i \in B_u, j: \rw(j) \ne i} \bE[p_{i, j}] \\
        & \ge \max\left\{A, \sum_{i \in B_v} \bE[p_i] + \sum_{i \in B_u, j: \rw(j) \ne i} \bE[p_{i, j}]\right\} + B.
    \end{align*}
    Rearranging the sums within the $\max$, we have that
    \begin{align*}
        \sum_{i \in B_v} \bE[p_i] + \sum_{i \in B_u, j: \rw(j) \ne i} \bE[p_{i, j}] & = \sum_i \bE[p_i] - \sum_{i \in B_u, j: \rw(j) = i} \bE[p_{i, j}] \\
        & = \sum_{i, j: \rw(j) \in B_v} \bE[p_{i, j}] + \sum_{i, j: \rw(j) \in B_u \text{ and } \rw(j) \ne i} \bE[p_{i, j}] \\
        & = C + D.
    \end{align*}
    Substituting in the $\max$ gives the bound to be proved.
\end{proof}

We are now ready to prove our main PoA lower bound.

\begin{lemma}
\label{lem:ub}
    In a mixed autobidding world with both utility maximizers and value maximizers, the PoA of the first-price auction has the following lower bound:
    \[
        \inf_{n, m, B_u, B_v, \{v_{i, j}\}} \poa(n, m, B_u, B_v, \{v_{i, j}\}) \ge \min_{t \in [0, 1]} \frac{1 + t \ln t}{2 - t + t \ln t} \approx 0.457.
    \]
\end{lemma}
\begin{proof}
    Fix any $n$, $m$, $B_u$, $B_v$, and $\{v_{i, j}\}$.
    Consider any bidding strategies $\{\b_i\}_i$ that form an equilibrium.
    Let
    \[
        V_1 = \sum_{j: \rw(j) \in B_v} \max_i v_{i, j}, \quad V_2 = \sum_{j: \rw(j) \in B_u} \max_i v_{i, j}.
    \]
    By Lemma~\ref{lem:combination}, we have 
    \[
        \frac{\sum_i \bE[\val_i]}{\sum_j \max_i v_{i, j}} \ge \frac{\max\{A, C + D\} + B}{V_1 + V_2}.
    \]

    By Lemma~\ref{lem:value}, we have $A + C \ge V_1$. So there must be some $x \in [0, 1]$, such that $A \ge x \cdot V_1$ and $C \ge (1 - x) \cdot V_1$.
    As for $B$ and $D$, we apply Lemma~\ref{lem:local}.
    For each $j$ where $\rw(j) \in B_u$, Lemma~\ref{lem:local} states that there exists some $y_j \in [0, 1]$, such that
    \[
        \bE[\val_{\rw(j), j}] = y_j \cdot \max_i v_{i, j}, \quad \sum_{i \ne \rw(j)} \bE[p_{i, j}] \ge (1 - y_j + y_j \ln y_j) \cdot \max_i v_{i, j}.
    \]
    Let $y = B/ V_2$. Since $B = \sum_{j: \rw(j) \in B_u} \bE[\val_{\rw(j), j}] \leq \sum_{j: \rw(j) \in B_u} \max_i v_{i, j} = V_2$, we have $y  \in [0, 1]$. In particular,
    \[
       y = \frac{\sum_{j: \rw(j) \in B_u} \bE[\val_{\rw(j), j}]}{\sum_{j: \rw(j) \in B_u} \max_i v_{i, j}} = \frac{\sum_{j: \rw(j) \in B_u} y_j \cdot \max_i v_{i, j}}{\sum_{j: \rw(j) \in B_u} \max_i v_{i, j}}.
    \]
    Since $t \mapsto 1 - t + t \ln t$ is convex, by Jensen's inequality, we must have
    \begin{align*}
        D & = \sum_{j: \rw(j) \in B_u} \sum_{i \ne \rw(j)} \bE[p_{i, j}] \ge \sum_{j: \rw(j) \in B_u} (1 - y_j + y_j \ln y_j) \cdot \max_i v_{i, j} \\
        & \ge (1 - y + y \ln y) \sum_{j: \rw(j) \in B_u} \max_i v_{i, j} = (1 - y + y \ln y) \cdot V_2.
    \end{align*}
    Plugging in the above derived lower bounds for $A$, $B$, $C$, and $D$, we have
    \[
        \frac{\sum_i \bE[\val_i]}{\sum_j \max_i v_{i, j}} \ge \frac{\max\{A, C + D\} + B}{V_1 + V_2} \ge \frac{\max\{x \cdot V_1, (1 - x) \cdot V_1 + (1 - y + y \ln y) \cdot V_2\} + y \cdot V_2}{V_1 + V_2},
    \]
    where $V_1, V_2 \ge 0$, $V_1 + V_2 > 0$ and $x, y \in [0, 1]$.
    In other words, we must have
    \begin{gather*}
        \poa(n, m, B_u, B_v, \{v_{i, j}\}) \ge \inf_{V_1, V_2 \ge 0, V_1 + V_2 > 0, x, y \in [0, 1]} \frac{\max\{x \cdot V_1, (1 - x) \cdot V_1 + (1 - y + y \ln y) \cdot V_2\} + y \cdot V_2}{V_1 + V_2}.
    \end{gather*}
    A careful analysis (see Lemma~\ref{lem:max} in Appendix~\ref{app:proofs}) of the right-hand-side shows that it equals to
    \[
        \min_{y \in [0, 1]} \frac{1 + y \ln y}{2 - y + y \ln y} \approx 0.457,
    \]
    which is achieved when $x = 1$ and $V_1 = (1 - y + y \ln y) \cdot V_2$.
    This concludes the proof.
\end{proof}

\subsection{A Matching Upper Bound}

To conclude this section, we present a problem instance in which our PoA lower bound can be actually achieved, even with $n = 2$ bidders and $m = 2$ auctions.

\begin{lemma}
\label{lem:lb}
    When $n = m = 2$, there exists $B_u$, $B_v$ and $\{v_{i, j}\}$ such that
    \[
        \poa(n, m, B_u, B_v, \{v_{i, j}\}) \le \min_{t \in [0, 1]} \frac{1 + t \ln t}{2 - t + t \ln t} \approx 0.457.
    \]
\end{lemma}
\begin{proof}
    Let $B_u = \{1\}$ and $B_v = \{2\}$.
    We choose $v_{1, 1} = 1$ and $v_{1, 2} = v_{2, 1} = 0$, and leave $v_{2, 2} > 0$ to be fixed later.
    The idea is to let bidder $2$ win in auction $2$ for free, and therefore bidder $2$ gains enough slackness in the ROI constraint and can therefore compete with bidder $1$ in auction $1$.
    To be specific, in order for bidder $2$ to win for free in auction $2$, we choose $b_{1, 2} = b_{2, 2} = 0$ so that bidder $2$ wins by tiebreaking and pays $0$.
    We want bidder $1$'s value-payment frontier in auction $1$ to have the following shape: the curve starts at $(0, t)$ for some $t \in [0, 1]$ (to be fixed later), and then increases linearly with slope precisely $1$, until it hits the upper bound, i.e., when the $y$-axis becomes $v_{1, 1} = 1$.
    This ensures that bidder $1$ bidding $b_{1, 1} = 0$ deterministically is a best response.
    Translating this to bidder $2$'s bidding strategy, we would like that for any $b \ge 0$,
    \[
        \Pr[b_{2, 1} \le b] = \min\left\{1, \frac{t}{1 - b}\right\}.
    \]
    One can show that the expected payment that bidder $2$ makes in auction $1$ is $1 - t + t \ln t$.
    In order to satisfy bidder $2$'s ROI constraint, we choose $v_{2, 2} = 1 - t + t \ln t$.
    One may check that the above bidding strategies in fact form an equilibrium.
    Note that $\bE[\val_1] = \Pr[b_{2, 1} = 0] \cdot v_{1, 1} = t$, and $\bE[\val_2] = v_{2, 2} = 1 - t + t \ln t$.

    Considering the ratio between the optimal welfare and the welfare resulting from the above equilibrium, we have
    \[
        \poa(n, m, B_u, B_v, \{v_{i, j}\}) \le \frac{\bE[\val_1 + \val_2]}{v_{1, 1} + v_{2, 2}} = \frac{1 + t \ln t}{2 - t + t \ln t},
    \]
    for any $t \in [0, 1]$.
    Minimizing over $t$ gives
    \[
        \poa(n, m, B_u, B_v, \{v_{i, j}\}) \le \min_{t \in [0, 1]} \frac{1 + t \ln t}{2 - t + t \ln t}. \qedhere
    \]
\end{proof}

Theorem~\ref{thm:main} is a direct corollary of Lemma~\ref{lem:ub} and Lemma~\ref{lem:lb}.

\section{Improved Efficiency with ML Advice}

In this section, we turn to a more practically-motivated setting, where the seller can rely on machine-learned advice to set reserves and thereby further improve the efficiency of the first-price auction.
In particular, we show that as the advice gets more and more accurate, there is a way to set reserves such that the PoA of the first-price auction approaches $1$.

\paragraph{Setup.}
Following prior work \cite{BalseiroDMMZ21},
we consider the setting in which for each auction $j$, the seller can access a signal $s_j$, which approximates the highest value $\max_i v_{i, j}$ in this auction.
Formally, we assume there is a constant $\gamma \in [0, 1]$ such that $s_j \in [\gamma \cdot \max_i v_{i, j}, \max_i v_{i, j}]$ for each auction $j$.\footnote{As we focus on the single-slot environment in this paper, for ease of presentation, our assumption here is weaker than \citet{BalseiroDMMZ21} as they assume for each bidder $i$ and each auction $j$, there is a signal $s_{i,j} \in [\gamma \cdot v_{i, j}, v_{i, j}]$, and setting $s_j:=\max_i s_{i,j}$ gives the signals which can be used in our setting. Our result can be extended to the multi-slot environment by adopting the assumption from \citet{BalseiroDMMZ21}.}


\paragraph{First-price auction with machine-learned reserves.} The way we exploit machine-learned signals is to use them as reserve prices in the first-price auction.
To be more specific, in each auction $j$, we set a reserve $r_j = s_j$, such that a bidder $i$ can win only if $i$'s bid $b_{i, j} \ge r_j$.
In particular, if $\max_i b_{i, j} < r_j$, then no one wins in auction $j$. This is different from the first-price auction without reserves, where there is always exactly $1$ winner.

We overload the notations and let $x_{i, j} = x_{i, j}(r_j, \{b_{i', j}\}_{i'})$ be the indicator variable that $i$ wins in auction $j$, and define $\{\val_{i, j}\}$ and $\{p_{i, j}\}$ in a similar way such that all these variables depend on both the reserves $\{r_j\}$ and the bidding profile $\{b_{i', j}\}_{i'}$.
For any $n$, $m$, $B_u$, $B_v$, and $\{v_{i, j}\}$, we consider the following parametrized instance-wise PoA when each $r_j = s_j \in [\gamma \cdot \max_i v_{i, j}, \max_i v_{i, j}]$:
\[
    \poa_\gamma(n, m, B_u, B_v, \{v_{i, j}\}) = \inf_{\substack{\{b_{i, j}\} \text{ form an equilibrium} \\ \forall j, r_j \in [\gamma \cdot \max_i v_{i, j}, \max_i v_{i, j}]}} \frac{\sum_i \bE[\val_i]}{\sum_j \max_i v_{i, j}}.
\]

Our main result in the section is captured by the following theorem.
\begin{theorem}
\label{thm:main-ml}
    For any $n$, $m$, $B_u$, $B_v$, $\{v_{i, j}\}$ and $\gamma \in [0, 1]$,
    \[
        \poa_\gamma(n, m, B_u, B_v, \{v_{i, j}\}) \ge \min_{t \in [0, 1]} \frac{1 + t \ln t - \gamma t (1 + \ln t)}{2 - t - \gamma + (1 - \gamma) t \ln t}.
    \]
\end{theorem}

One way to interpret the lower bound is it provides a smooth transition between the case without machine-learned advice (or with totally unreliable advice, i.e., $\gamma = 0$), and the case where the advice is perfectly accurate (i.e., $\gamma = 1$) and the seller knows precisely the value of the rightful winner in each auction.
In the former scenario corresponding to $\gamma = 0$, we recover the tight PoA bound stated in Theorem~\ref{thm:main}.
In the latter scenario corresponding to $\gamma = 1$, as one would expect, the theorem gives a lower bound of $1$, which means the auctions are perfectly efficient.
More generally, for $\gamma$ between $0$ and $1$, the lower bound in Theorem~\ref{thm:main-ml} increases monotonically as $\gamma$ increases, since the fraction within the $\min$ increases in $\gamma$ for any fixed $t \in [0, 1]$: 
\[
  \frac{1 + t \ln t - \gamma t (1 + \ln t)}{2 - t - \gamma + (1 - \gamma) t \ln t} = \frac{(1-t)^2}{(2-t+t\ln t)(1+t\ln t) - (1+t\ln t)^2\gamma} + \frac{t+ t\ln t}{1+t\ln t},
\]
where $1+t\ln t \in (0, 1]$ and $2-t+t\ln t \geq 1 \geq \gamma \geq \gamma \cdot (1+t\ln t)$.

In order words, the efficiency of the first-price auction with machine-learned reserves improves as the reserves become more and more accurate.
To better illustrate the rate at which the efficiency improves, we plot our lower bound against the accuracy $\gamma$ in Figure~\ref{fig:ml}.
As the plot shows, our lower bound improves approximately linearly as the accuracy $\gamma$ increases from $0$ to $1$.

\begin{figure}
\centering
\includegraphics[width=0.8\linewidth]{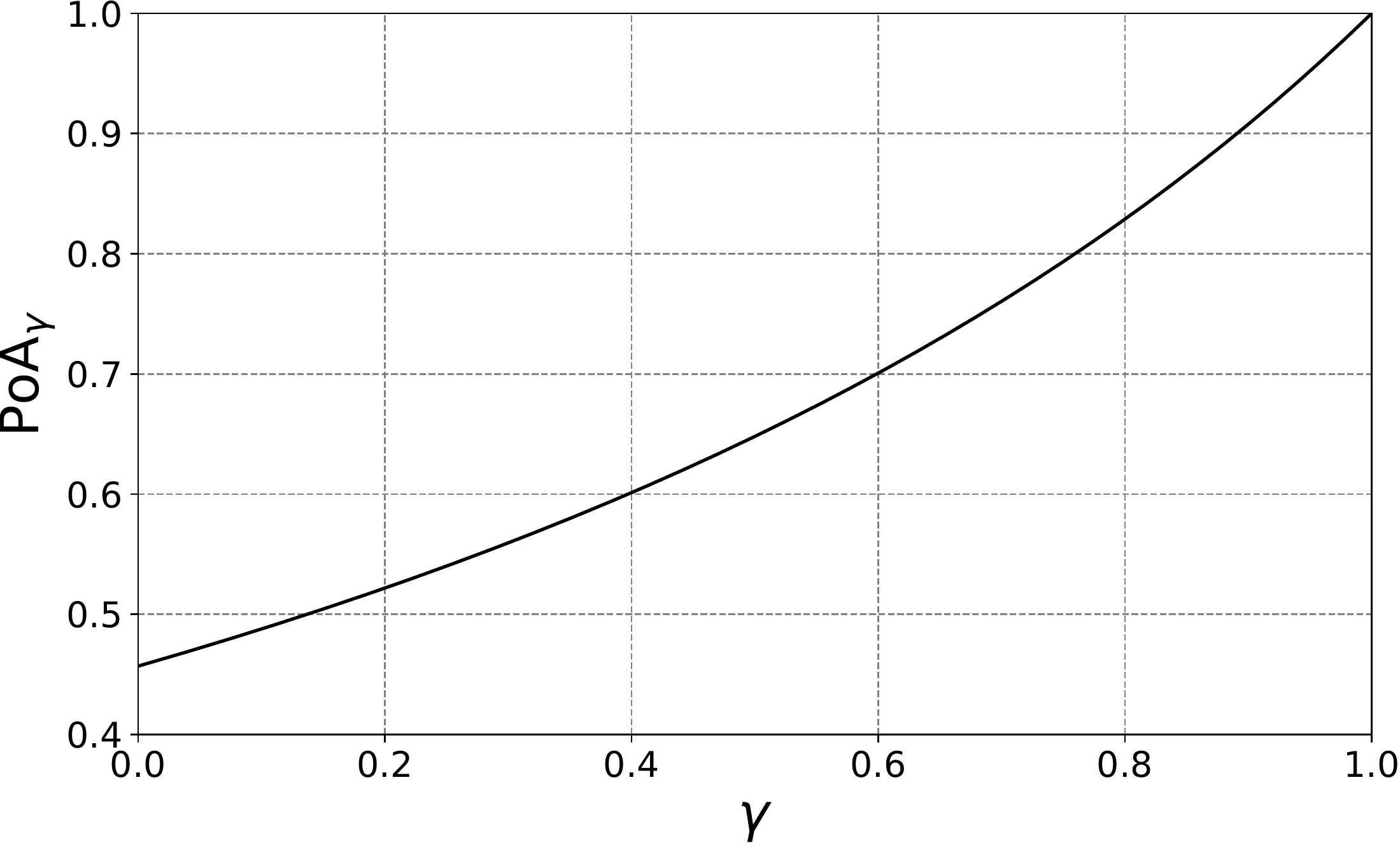}
\caption{The efficiency of the first-price auction with machine-learned reserves when the accuracy of the reserves is $\gamma \in [0, 1]$.}
\label{fig:ml}
\end{figure}

To prove Theorem~\ref{thm:main-ml}, the plan is to apply Lemma~\ref{lem:combination}, which provides a lower bound on the welfare in any equilibrium, and consider tradeoffs between the $4$ quantities involved in the lower bound.
The key difference is that we obtain better tradeoffs with machine-learned reserves.
In the rest of the section, we fix $n$, $m$, $B_u$, $B_v$, and $\{v_{i, j}\}$.
With machine-learned reserves, the tradeoff given by Lemma~\ref{lem:value} improves to the following (note the factor of $(1 - \gamma)$ on the left hand side).

\begin{lemma}
\label{lem:value-ml}
    In auctions with machine-learned reserves of accuracy $\gamma \in [0, 1]$, for any $\{\b_i\}_i$ that form an equilibrium,
    \[
        (1 - \gamma) \cdot \sum_{i \in B_v} \bE[\val_i] + \sum_{i, j: \rw(j) \in B_v} \bE[p_{i, j}] \ge \sum_{j: \rw(j) \in B_v} \max_i v_{i, j}.
    \]
\end{lemma}
\begin{proof}
    The proof is similar to that of Lemma~\ref{lem:value} at a high level, except that the expected payment is higher, because in each auction $j$, the rightful winner $\rw(j)$ must bid at least $r_j$ (which is a constant fraction of $v_{\rw(j), j}$) unless the highest competing bid is at least $v_{\rw(j), j}$ with probability $1$.
    
    Again fix any equilibrium strategies $\{\b_i\}$.
    By the same argument as in the proof of Lemma~\ref{lem:value}, for any value maximizer $i \in B_v$, we have
    \[
        \bE[\val_i] \ge \sum_{j: i = \rw(j)} v_{i, j} \cdot \Pr\left[\max_{i' \ne i} b_{i', j} \le v_{i, j}\right].
    \]
    Also, observe that for any value maximizer $i \in B_v$ and auction $j$ where $\rw(j) = i$, at least one of the following properties must hold: Either (1) $b_{i, j} \ge r_{i, j}$ deterministically, or (2) $\Pr[\max_{i' \ne i} b_{i', j} \ge v_{i, j}] = 1$. 
    This is because if neither of the two holds, then $i$ for any $b < r_j$ in the support of $b_{i, j}$, $i$ would be better off by moving probability mass from $b$ to $v_{i, j}$, which preserves feasibility and strictly increases the value $i$ receives from $j$.
    For each $j$ where (1) holds, we have
    \begin{align*}
        \sum_{i'} \bE[p_{i', j}] & \ge v_{i, j} \cdot \left(1 - \Pr\left[\max_{i' \ne i} b_{i', j} \le v_{i, j}\right]\right) + r_j \cdot \Pr\left[\max_{i' \ne i} b_{i', j} \le v_{i, j}\right] \\
        & \ge v_{i, j} \cdot \left(1 - \Pr\left[\max_{i' \ne i} b_{i', j} \le v_{i, j}\right]\right) + \gamma \cdot v_{i, j} \cdot \Pr\left[\max_{i' \ne i} b_{i', j} \le v_{i, j}\right] \\
        & = v_{i, j} - (1 - \gamma) \cdot v_{i, j} \cdot \Pr\left[\max_{i' \ne i} b_{i', j} \le v_{i, j}\right].
    \end{align*}
    For each $j$ where (2) holds, we immediately know that
    \begin{align*}
        \sum_{i'} \bE[p_{i', j}] & \ge v_{i, j} \cdot \Pr\left[\max_{i' \ne i} b_{i', j} \ge v_{i, j}\right] = v_{i, j} \ge v_{i, j} - (1 - \gamma) \cdot v_{i, j} \cdot \Pr\left[\max_{i' \ne i} b_{i', j} \le v_{i, j}\right].
    \end{align*}
    So combining the two cases and summing over $j$ where $i = \rw(j)$,
    \[
        \sum_{i', j: \rw(j) = i} \bE[p_{i', j}] \ge \sum_{j: \rw(j) = i} \left(v_{i, j} - (1 - \gamma) \cdot v_{i, j} \cdot \Pr\left[\max_{i' \ne i} b_{i', j} \le v_{i, j}\right]\right).
    \]
    Combining this with the lower bound on $\bE[\val_i]$ above, we get
    \[
        (1 - \gamma) \cdot \bE[\val_i] + \sum_{i', j: \rw(j) = i} \bE[p_{i', j}] \ge \sum_{j: \rw(j) = i} v_{i, j}.
    \]
    Further summing over $i \in B_v$ gives the inequality to be proved.
\end{proof}

We note that the above lemma also implies a lower bound of $1 / (2 - \gamma)$ on the PoA of the first-price auction with machine-learned reserves in the full autobidding world in which all bidders are value maximizers.
When $\gamma \to 0$ (i.e., reserves are extremely unreliable), the bound degrades to $1/2$; when $\gamma \to 1$ (i.e., reserves are extremely accurate), the bound approaches $1$, which means the auctions are efficient.

\begin{corollary}
\label{cor:main-ml-autoonly}
    For any $n$, $m$, $\{v_{i, j}\}$ and $\gamma \in [0, 1]$, $\poa_\gamma(n, m, B_u = \emptyset, B_v = [n], \{v_{i, j}\}) \ge 1 / (2 - \gamma)$.
\end{corollary}

For utility maximizers, we adapt the local analysis, and obtain the following improved tradeoff between value and payment.

\begin{lemma}
\label{lem:local-ml}
    In auctions with machine-learned reserves of accuracy $\gamma \in [0, 1]$, for any $\{\b_i\}_i$ that form an equilibrium, and any bidder $i$ and auction $j$ where $\rw(j) = i$, there exists some $x \in [0, 1]$ such that
    \[
        \bE[\val_{i, j}] = x \cdot v_{i, j} \quad \text{and} \quad \sum_{i' \ne i} \bE[p_{i', j}] \ge (1 - x + (1 - \gamma) x \ln x) \cdot v_{i, j}.
    \]
\end{lemma}
\begin{proof}
    The first part of the reasoning is almost identical to the proof of Lemma~\ref{lem:local}: For any $b$ in the support of $b_{i, j}$, Lemma~\ref{lem:marginal} gives us a way to upper bound the CDF of the highest other bid, and subsequently to lower bound the payment made by other bidders.
    Without repeating the entire argument, we start directly from this lower bound.
    Let $F$ be the CDF of the highest other bid, i.e., for any $t \ge 0$,
    \[
        F(t) = \Pr\left[\max_{i' \ne i} b_{i', j} \le t\right].
    \]
    When $b \ge v$, we have
    \[
        \sum_{i' \ne i} \bE[p_{i', j} \mid b_{i, j} = b] \ge (1 - F(b)) \cdot v \ge (1 - F(b) + (1 - \gamma) \cdot F(b) \ln F(b)) \cdot v_{i, j}.
    \]
    When $b < v$, again we must have one of the following two: Either (1) $b \ge r_j \ge \gamma \cdot v_{i, j}$, or (2) $F(v_{i, j}) = 1$.
    When (1) holds, we have
    \[
        \sum_{i' \ne i} \bE[p_{i', j} \mid b_{i, j} = b] \ge (1 - F(b)) \cdot v_{i, j} + F(b) \ln F(b) \cdot (v_{i, j} - b) \ge (1 - F(b) + (1 - \gamma) F(b) \ln F(b)) \cdot v_{i, j}.
    \]
    When (2) holds, clearly we have
    \[
        \sum_{i' \ne i} \bE[p_{i', j} \mid b_{i, j} = b] \ge v_{i, j} \cdot F(v_{i, j}) = v_{i, j} \ge (1 - F(b) + (1 - \gamma) F(b) \ln F(b)) \cdot v_{i, j}.
    \]
    Combining these cases and taking the expectation over $b_{i, j}$ (again letting $x = \bE[F(b_{i, j})]$), we get
    \[
        \sum_{i' \ne i} \bE[p_{i', j}] \ge \bE[(1 - F(b_{i, j}) + (1 - \gamma) \cdot F(b_{i, j}) \ln F(b_{i, j})) \cdot v_{i, j}] \ge (1 - x + (1 - \gamma) \cdot x \ln x) \cdot v_{i, j},
    \]
    where the last inequality from Jensen's inequality since $z \mapsto 1 - z + (1 - \gamma) \cdot z \ln z$ is convex on $[0, 1]$ for any $\gamma \in [0, 1]$.
\end{proof}

We are now ready to plug Lemma~\ref{lem:value-ml} and Lemma~\ref{lem:local-ml} into Lemma~\ref{lem:combination} and prove Theorem~\ref{thm:main-ml}.

\begin{proof}[Proof of Theorem~\ref{thm:main-ml}]
    Reproducing the argument in the proof of Lemma~\ref{lem:ub} (but with the improved tradeoffs given by Lemma~\ref{lem:value-ml} and Lemma~\ref{lem:local-ml}), we can lower bound $\poa_\gamma(n, m, B_u, B_v, \{v_{i, j}\})$ by
    \[
         \inf_{V_1, V_2 \ge 0, V_1 + V_2 > 0, x, y \in [0, 1]} \frac{\max\{x \cdot V_1, (1 - (1 - \gamma)x) \cdot V_1 + (1 - y + (1 - \gamma) y \ln y) \cdot V_2\} + y \cdot V_2}{V_1 + V_2}.
    \]
    A careful analysis (see Lemma~\ref{lem:max-ml} in Appendix~\ref{app:proofs}) of the infimum shows that it equals to
    \[
        \min_{y \in [0, 1]} \frac{1 + y \ln y - \gamma y (1 + \ln y)}{2 - y - \gamma + (1 - \gamma) y \ln y}. 
    \]
    \qedhere
\end{proof}

\appendix

\section{Analysis of the Infimum}
\label{app:proofs}

\begin{lemma}
\label{lem:max}
    \[
        \inf_{V_1, V_2 \ge 0, V_1 + V_2 > 0, x, y \in [0, 1]} \frac{\max\{x \cdot V_1, (1 - x) \cdot V_1 + (1 - y + y \ln y) \cdot V_2\} + y \cdot V_2}{V_1 + V_2} = \min_{y \in [0, 1]} \frac{1 + y \ln y}{2 - y + y \ln y}.
    \]
\end{lemma}
\begin{proof}
    Without loss of generality, suppose $V_1 + V_2 = 1$, so the domain is compact, and we only need to prove that
    \[
      \inf_{V_1 + V_2 = 1, V_1, V_2, x, y \in [0, 1]} \frac{\max\{x \cdot V_1, (1 - x) \cdot V_1 + (1 - y + y \ln y) \cdot V_2\} + y \cdot V_2}{V_1 + V_2} = \min_{y \in [0, 1]} \frac{1 + y \ln y}{2 - y + y \ln y}.
    \]
    
    Since the term inside the infimum is continuous on the domain we can replace the $\inf$ with a $\min$.
    In other words, there exists $V_1^*$, $V_2^*$, $x^*$ and $y^*$ that achieve the minimum.
    If $V_2 = 0$, then clearly it is lower bounded by $1/2$ (which is clearly suboptimal).
    So we must have $V_2^* > 0$.

    Now observe that the two terms within the inner $\max$ must be equal when the minimum is achieved.
    This is because: if $V_1^* \ge (1 - y^* + y^* \ln y^*) \cdot V_2^*$, then $x^*$ must equalize the two terms (any other choice of $x^*$ would be suboptimal).
    Otherwise, the inner $\max$ must be $(1 - x^*) \cdot V_1^* + (1 - y^* + y^* \ln y^*) \cdot V_2^*$, which means $x^*$ must be $1$, in which case $V_1^*$ does not appear at all.
    But then, increasing $V_1^*$ and decreasing $V_2^*$ simultaneously would strictly decrease the value, which contradicts the optimality of these choices.
    So we must have
    \[
        x^* \cdot V_1^* = (1 - x^*) \cdot V_1^* + (1 - y^* + y^* \ln y^*) \cdot V_2^* \implies V_1^* = \frac{(1 - y^* + y^* \ln y^*) \cdot V_2^*}{2x^* - 1} \text{ and } x^* \in (1/2, 1].
    \]
    Plugging this in, the minimum simplifies to
    \[
        \frac{x^* \cdot (1 + y^* + y^* \ln y^*) - y^*}{2x^* - y^* + y^* \ln y^*}.
    \]

    Observe that the denominator of the above fraction is always positive whenever $x^* \in (1/2, 1]$ and $y^* \in [0, 1]$.
    Moreover, depending on the value of $y^*$, the partial derivative of the above fraction with respect to $x^*$ is either never $0$ or always $0$.
    In the former case, it must be the case that $x^* = 1$ (any other choice cannot be optimal); in the latter case, the value of $x^*$ does not matter, so without loss of generality we can choose $x^* = 1$.
    Plugging this in, the minimum simplifies to
    \[
        \frac{1 + y^* \ln y^*}{2 - y^* + y^* \ln y^*} = \min_{y \in [0, 1]} \frac{1 + y \ln y}{2 - y + y \ln y}.
    \]
    This concludes the proof.
\end{proof}

\begin{lemma}
\label{lem:max-ml}
    \begin{align*}
        &~\inf_{V_1, V_2 \ge 0, V_1 + V_2 > 0, x, y \in [0, 1]} \frac{\max\{x \cdot V_1, (1 - (1 - \gamma)x) \cdot V_1 + (1 - y + (1 - \gamma) y \ln y) \cdot V_2\} + y \cdot V_2}{V_1 + V_2} \\
        =&~ \min_{y \in [0, 1]} \frac{1 + y \ln y - \gamma y (1 + \ln y)}{2 - y - \gamma + (1 - \gamma) y \ln y}.
    \end{align*}
\end{lemma}
\begin{proof}
    When $\gamma = 1$, it is easy to check the above minimum is $1$.
    In the rest of the proof, we assume $\gamma \in [0, 1)$.
    We break this into two cases:
    \begin{itemize}
        \item When
        \[
            V_1 \le \gamma \cdot V_1 + (1 - y + (1 - \gamma) y \ln y) \cdot V_2 \implies V_1 \le \frac{(1 - y + (1 - \gamma) y \ln y) \cdot V_2}{1 - \gamma},
        \]
        the minimum becomes
        \[
            \min \frac{(1 - (1 - \gamma)x) \cdot V_1 + (1 + (1 - \gamma) y \ln y) \cdot V_2}{V_1 + V_2},
        \]
        and the minimum must be achieved when $x = 1$, in which case it further simplifies to
        \[
            \min \frac{\gamma \cdot V_1 + (1 + (1 - \gamma) y \ln y) \cdot V_2}{V_1 + V_2}.
        \]
        Since $y \ln y \in [-1/e, 0]$ and $(1 + (1 - \gamma) y \ln y) > \gamma$, we want $V_1$ to be as large as possible, which means $V_1 = (1 - y + (1 - \gamma) y \ln y) \cdot V_2 / (1 - \gamma)$.
        Plugging this in, the minimum becomes
        \[
            \min \frac{1 + y \ln y - \gamma y (1 + \ln y)}{2 - y - \gamma + (1 - \gamma) y \ln y}.
        \]
        \item When
        \[
            V_1 \ge \gamma \cdot V_1 + (1 - y + (1 - \gamma) y \ln y) \cdot V_2 \implies V_1 \ge \frac{(1 - y + (1 - \gamma) y \ln y) \cdot V_2}{1 - \gamma},
        \]
        the minimum must be achieved when the two terms within the inner $\max$ are equalized, i.e.,
        \[
            x \cdot V_1 = (1 - (1 - \gamma)x) \cdot V_1 + (1 - y + (1 - \gamma) y \ln y) \cdot V_2 \implies V_1 = \frac{(1 - y + (1 - \gamma) y \ln y) \cdot V_2}{(2 - \gamma) x - 1}.
        \]
        This also means we can restricted the domain of $x$ to $x \in (1 / (2 - \gamma), 1]$.
        Plugging this in, the minimum simplifies to
        \[
            \min \frac{x(1 + y - \gamma y + (1 - \gamma) y \ln y) - y}{(2 - \gamma)x - y + (1 - \gamma) y \ln y}.
        \]
        Observe that the denominator is always positive for $x$ and $y$ in the domain, so the derivative of the fraction with respect to $x$ is either never $0$ or always $0$.
        In the former case, the minimum can only be achieved when $x = 1$; in the latter case, the value of $x$ does not matter, so we can without loss of generality choose $x = 1$.
        Plugging $x = 1$ in, the minimum simplifies to
        \[
            \min \frac{1 + y \ln y - \gamma y (1 + \ln y)}{2 - y - \gamma + (1 - \gamma) y \ln y}.
        \]
    \end{itemize}
    Combining the two cases concludes the proof.
\end{proof}

\bibliographystyle{plainnat}
\bibliography{ref}

\end{document}